\documentclass[runningheads]{llncs}
\usepackage{graphicx} % necessary for inserting .pdfs
\usepackage[firstpage]{draftwatermark} % free badge placement

\usepackage{braket}
\usepackage{float}
\usepackage{graphicx}
\usepackage{hyperref}

\usepackage{amsthm, amsmath, amssymb}
\usepackage{mathrsfs}
\usepackage{makecell}
\usepackage{tikz}
\usepackage{qcircuit}
\usepackage{subfigure}
\usepackage{listings}
\usepackage{multirow}
\usepackage{appendix}

\usepackage{algorithm}
\usepackage{algorithmic}

\begin{document}
\title{QReach: A Reachability Analysis Tool for Quantum Markov Chains}

\author{Aochu Dai\inst{1} \and
Mingsheng Ying\inst{1,2}}
\authorrunning{A. Dai and M. Ying}
\institute{Department of Computer Science and Technology,
Tsinghua University, \\Beijing 100084, China\\ \email{dac22@mails.tsinghua.edu.cn} \and
Institute of Software, Chinese Academy of Sciences, Beijing 100190, China\\ \email{yingms@ios.ac.cn}}
\maketitle              % typeset the header of the contribution
%

% \SetWatermarkAngle{0}
% \SetWatermarkText{\raisebox{16.5cm}{%
% \hspace{0.1cm}%
% \href{https://doi.org/10.5281/zenodo.10939993}{\includegraphics{1-available}}%
% \hspace{9cm}%
% \includegraphics{2-functional}%
% }}
\begin{abstract} %We present the \textit{first automatic tool} for reachability analysis of quantum Markov chains. ... 

We present QReach, the first reachability analysis tool for quantum Markov chains based on decision diagrams CFLOBDD (presented at \textit{CAV} 2023). QReach provides a novel framework for finding reachable subspaces, as well as a series of model-checking subprocedures like image computation. Experiments indicate its practicality in verification of quantum circuits and algorithms. QReach is expected to play a central role in future quantum model checkers.

% The abstract should briefly summarize the contents of the paper in
% 15--250 words.

\keywords{Reachability analysis \and Quantum Markov chain \and Quantum Model Checking.}
\end{abstract}

\section{Introduction}
% \begin{itemize}
%     \item Background, motivations
%     \item Purpose of this paper
%     \item contribution 
% \end{itemize}
A rapid growth of quantum computing hardware has been witnessed in the last few years. As a recent breakthrough, IBM has introduced its new quantum processor Condor, which breaks the 1000-qubit barrier. Many researchers share the belief that quantum computation will be scalable and stable enough for some meaningful quantum algorithms in the foreseeable future. In the era of Fault-Tolerant Quantum Computing (FTQC), quantum systems may be too complicated to be designed and verified manually. On the other hand, systematic ventures occurring in a quantum circuit or a communication protocol may differ significantly from those in classical systems and may be counterintuitive. The success of model checking techniques in classical computing and communication industry motivates us to extend it for analysis and verification of various temporal properties of a quantum system. Indeed, several model checking algorithms have been proposed for quantum automata and quantum Markov chains \cite{mateus2009temporal,ying2021model,mateus2006weakly,reachrecursive}. Additionally, some basic communication protocols like BB84 have passed the verification of the proposed quantum model checkers \cite{baltazar2008quantum}. However, these quantum model checkers cannot be applied to larger quantum systems.  

As is well known, the scalability of classical model checkers heavily relies on the data structures (in particular, various DDs (Decision Diagrams), e.g. ROBDD) employed in them for representing the system under checking. 
Several quantum generalisations of DDs have been introduced for modelling, simulation, and verification of (combinational) quantum circuits, like QMDD \cite{niemann2015qmdds}, TDD \cite{hong2022tensor}, and LimDD \cite{vinkhuijzen2023limdd}, % QMDD, TDD, LimDD, CFLOBDD
providing different degrees of compression for quantum states and operators. Based on these diagram structures, some simulation or verification tools were developed for experimental tasks like equivalence checking and bug finding \cite{chen2023autoq,chen2023automata}. Decision diagrams have well-defined canonicity and regularization, which motivates us to implement quantum model-checking algorithms by means of DDs. Up to now, however, these quantum DDs have not been used in quantum model checking.

In this paper, we incorporate quantum DDs into quantum model checking for the first time. 
Quantum Markov chains (QMCs for short) have been adopted as a fundamental model of many quantum information processing systems (e.g. quantum communication protocols, semantics of quantum programs, etc). So, we choose to use QMCs as our system model. As is well-known, reachability analysis is a core task in classical model checking algorithms. In the quantum case, indeed, reachability analysis has been applied in quantum communication, quantum control and termination analysis of quantum programs among many others. Therefore, we focus on the issue of reachability analysis of quantum Markov chains. In addition, we decide to use Context-Free-Language Ordered Binary Decision Diagrams \cite{sistla2022cflobdds} (CFLOBDD for short), one of the most efficient quantum DDs as the backend of our tool to provide support for functionalities. We also refer to Quasimodo \cite{sistla2023symbolic}, a quantum circuit simulator based on CFLOBDD, for some of the code's implementation details. Supported by the efficiency of the DD representation, our tool is well scalable and has the potential to be expanded into large-scale quantum circuit model checkers in the future.

\textbf{Contributions of the paper}: This paper introduces the first reachability analysis tool for QMCs, called QReach\footnote{Available at \href{https://doi.org/10.5281/zenodo.10931240}{https://doi.org/10.5281/zenodo.10931240}}. It can efficiently compute reachable subspaces of QMCs with the following techniques: \begin{itemize}\item Subspaces of QMCs, which are usually defined by atomic propositions in Birkhoff-von Neumman quantum logic, are represented as CFLOBDDs in QReach. 
\item Partitioning and frontier set simplification are introduced in our algorithm to reduce the size of data structures, in analogy to corresponding techniques in classical symbolic model checking. \end{itemize}

\section{Quantum reachability analysis}
% \begin{itemize}
%     \item Some concepts of quantum Markov chain or concurrent quantum programs.
%     \item the significance of reachability analysis.
% \end{itemize}

%Quantum model checking has the potential to verify increasingly complex quantum systems. It can be used to verify quantum circuits, test a quantum system’s security, and develop new quantum hardware or software. Over the past decade, several algorithms and models have been proposed to solve or clarify quantum model-checking problems. Some temporal extensions of exogenous quantum propositional logic were proposed as quantum versions of computation tree logic and linear temporal logic. In \cite{}, various quantum models and corresponding model-checking algorithms are discussed, including quantum linear-time properties, reachability analysis with quantum graph theory, and techniques for model-checking a variety of quantum Markov chains.

% Based on past experiences with classical model-checking techniques, We believe that reachability analysis plays a crucial role in quantum model-checking methods.

For convenience of the reader, we briefly review the model of QMCs and their reachable subspaces. Recall that a Markov chain (MC for short) is a pair $\langle S, P\rangle$, where $S$ is a finite set of states and $P$ is a transition probability matrix $P:\ S\times S \rightarrow [0,1]$ satisfying a normalization condition  $\sum_{s^\prime\in S}P(s,s^\prime) = 1$ for any $s\in S$. Similarly, a QMC is defined as a pair $\langle\mathcal{H}, \mathcal{E}\rangle$, where $\mathcal{H}$ is the state Hilbert space of the quantum system under consideration and $\mathcal{E}$ is a quantum operation describing the evolution of the system, i.e. a mapping from a quantum state $\rho$ to another $\mathcal{E}(\rho)$ (also called a quantum channel in quantum information literature). 
Table-\ref{tab1} gives a detailed comparison between classical MCs and QMCs:\begin{table}
\centering
\caption{Classical Markov chains vs quantum Markov chains.}\label{tab1}
\begin{tabular}{|l|c|c|}
\hline
 &  (Discrete-time) Markov Chain & Quantum Markov Chain\\
\hline
State space & Finite or countable set & \thead{Finite-dimensional or\\ separable Hilbert space}\\
\hline
Initialization & \thead{Probability  distribution: \\ $\iota_{init} : S\rightarrow [0,1]$ \\ $\sum_{s\in S}\iota_{init}(s) = 1$} 
               & \thead{Density matrix: \\ $\rho = \sum_{i} p_i|\psi_i\rangle\langle\psi_i|$ \\ $\sum p_i = 1$} \\
\hline
Transition & \thead{Probability transition matrix: \\ $P : S\times S\rightarrow [0,1]$ \\ $\sum_{s^\prime\in S}P(s,s^\prime) = 1$} 
               & \thead{Quantum operation: \\ $\mathcal{E}(\rho) = \sum_i E_i \rho E_i^\dag$ \\ $\sum E_i^\dag E_i = I$} \\
\hline
Logic & Probabilistic temporal logic & \thead{Temporal extension of \\ Birkhoff-von Neumann quantum Logic} \\
\hline
% Techniques & & \\
% \hline
\end{tabular}
\end{table}
% The quantum Markov chain is a quantum generalization of its classical counterpart: In a quantum Markov chain, a Hilbert space $\mathcal{H}$ serves as the state space, and a super operator $\mathcal{E}$ acting on it serves as a transition matrix. A quantum Markov chain can model many essential quantum protocols and algorithms due to the flexibility of the definition of channel $\mathcal{E}$, like random quantum walk, repeat-until-success circuits, and quantum Bernoulli factories.
\begin{itemize}\item A pure state of an $n$-dimensional quantum system is represented by a unit complex vector $|\psi\rangle\in\mathbb{C}^n$. In Table-\ref{tab1}, the density operator $\rho=\sum_ip_i|\psi_i\rangle\langle\psi_i|$ is a mathematical representation of  ensemble $\{(p_i,|\psi_i\rangle\}$, meaning the system is in state $|\psi_i\rangle$ with probability $p_i$. Thus, $\rho$ can be seen as a quantum analog of the initial probability distribution in a classical MC.
%if we view states $s\in S$ as a computational basis of $\mathcal{H}$, forming a diagonal density matrix. 
\item The quantum operation $\mathcal{E}$ is a quantum generalization of the transition probability matrix in a classical MC. According to the principles of quantum mechanics, it can be mathematically modelled as $\mathcal{E}(\rho) = \sum_i E_i \rho E_i^\dag$,
where each $E_i$ is an $n\times n$ complex matrices (called \textit{Kraus matrices}) satisfying the condition $\sum E_i^\dag E_i = I$ (the unit matrix), which is a counterpart of the normalization condition in a classical MC. In particular, the evolution of a  closed quantum system is described by $\mathcal{E}(\rho)=U\rho U^\dag$, where $U$ is a unitary matrix, i.e. $UU^\dag=U^\dag U=I$. 
\end{itemize}

\textbf{Quantum reachability problem}: Given a QMC $\mathcal{C} = \langle\mathcal{H}, \mathcal{E}\rangle$ and any initial state $\rho$ in $\mathcal{H}$, compute the reachable subspace:
\begin{align}
    \label{reach-def}
    \mathcal{R}_\mathcal{C}(\rho) = \text{span}\{\ket{\psi}\in \mathcal{H} : \ket{\psi} \text{ is reachable from }\rho \text{ in } \mathcal{C}\}
\end{align}
Intuitively, $\mathcal{R}_\mathcal{C}(\rho)$ consists of not only the states reached in the execution path  $\rho\rightarrow \mathcal{E}(\rho)\rightarrow\mathcal{E}^2(\rho)\rightarrow\cdots\rightarrow\mathcal{E}^n(\rho)\rightarrow\cdots$ but also their linear combinations. 

\begin{example}[Quantum random walk]
\label{example1}
Consider a quantum random walk on a 4-cycle~\cite{venegas2012quantum,sequential} with the Hadamard operator as a coin $c$, shown in Figure-\ref{fig:qwalk}. The walking space is a 4-dimensional Hilbert space $\mathcal{H}_p$ on the bottom two qubits $p_1,p_2$, supported by four position states $\{|0\rangle_p, |1\rangle_p, |2\rangle_p, |3\rangle_p\}$ in the computational basis. After applying the coin $H$ in each step, the evolution of the system is described by a quantum conditional (shift operator):
%$S$ which has the form:
$$S = |0\rangle_c\langle 0|\otimes\sum_i |i+1\rangle_p\langle i| + |1\rangle_c\langle 1|\otimes\sum_i |i-1\rangle_p\langle i|$$
where the neighbor-state $|i+1\rangle$ and $|i-1\rangle$ are computed modulo 4. It means that the walker can simultaneously walk in different directions, which is the main difference between classic and quantum random walks.\begin{figure}[htbp] \centering
\centerline{
\Qcircuit @C=1.0em @R=0.6em {
\lstick{d}    & \qw & \qw & \qw & \qw   & \qw      & \qw                       & \targ     & \qw   & \meter& \cw  \\
\lstick{c}    &     & \qw & \qw & \qw > & \gate{H} & \multigate{2}{\mathit{S}} & \qw       & \qw   & \qw   & \qw  \\
\lstick{p_1}  & \qwx&     & \qw & \qw > & \qw      & \ghost{\mathit{S}}        & \ctrl{-2} & \qw   & \qw   & \qwx\\
\lstick{p_2}  & \qwx& \qwx&     & \qw > & \qw      & \ghost{\mathit{S}}        & \ctrl{-3} & \qw   & \qwx  & \qwx \\
              & \qwx& \qwx& \qwx& \qw   & \qw      & \qw                       & \qw       & \qw \qwx & \qwx&\qwx \\
              & \qwx& \qwx& \qw & \qw   & \qw      & \qw                       & \qw       & \qw   & \qw \qwx&\qwx \\
              & \qwx& \qw & \qw & \qw   & \qw      & \qw                       & \qw       & \qw   & \qw & \qw \qwx
}
}
\caption{Quantum random walk on a 4-cycle. This figure is from~\cite{sequential},~\cite{chen2022}}\label{fig:qwalk}
% \hrulefill \vspace*{4pt}
\end{figure}
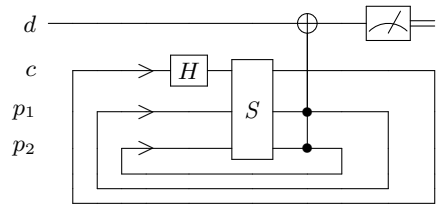
%Since multi-controlled gates and the measured qubit $d$ don't influence the quantum walk, we can model the quantum system composed of qubits $c$, $p_1$, and $p_2$ with a 
This system can be modelled as a QMC with $\mathcal{E}$ being defined by  the unitary operator $S(H_c\otimes I_p)$. Let it start in pure state $\rho = |000\rangle\langle000|$. Using the tool QReach presented in this paper, one can compute that the reachable space of this QMC is the $6$-dimensional space with linear independent basis $\{\ket{000}, \ket{001}+\ket{111}, \ket{100}-\ket{110}, \ket{101} +\ket{001}, \ket{010}, \ket{011}+\ket{101}\}$. It is surprising that not the whole $8$-dimensional state space $\mathcal{H}_c\otimes\mathcal{H}_p$ is  reachable although any position in $\mathcal{H}_p$ may be hit in some time.
\end{example}

\section{Architecture and Data structures}
% \begin{itemize}
%     \item Some representations of kraus operators. QASM files or other input instructions.
%     \item Represent projectors
%     \item Simulation-based algorithms
%     \item Introduction and modifications for CFLOBDD
%     \item Some more applications. For example, to check whether a state is in a subspace. Or set some boundary for reachable states calculations.
% \end{itemize}

% Paragraph for input, output, and applications.

In this section, we elaborate on the architecture of QReach and some reasoning techniques for quantum circuits based on the CFLOBDD backend. Although our target is specified on quantum reachability analysis, we believe that some functionalities of QReach are also useful for other tasks.
%, which are available to users for invoking and modifying ({\color{blue}http:?????}).

\subsection{Architecture of QReach}\label{QReach-arch}
An overview of the architecture of QReach is presented in Figure-\ref{fig:arch}. For convenience of presentation, we describe QReach's procedure with Example-\ref{example1}. Let the system starting in state $\rho=|000\rangle\langle000|$. To illustrate the capability of QReach for handling general quantum operations, we consider a faulty quantum random walk in which a bit-flip error happens in front of the Hadamard gate with probability $p$. The behavior of the system
%Example \ref{example1} 
is then modelled by $\mathcal{E}(\rho) = E_0 \rho E_0^\dag+E_1 \rho E_1^\dag$, where $E_0 = \sqrt{1-p}\ S (H_c\otimes I_p)$, and $E_1 = \sqrt{p}\ S (H_c\otimes I_p) (X\otimes I_p)$. Our purpose is to compute the reachable subspace of $\mathcal{E}$. An example Python program for this process is shown in Figure-\ref{fig:demo}.

% Show some other functions. More details, for example, what is the output, how to get the result
\begin{figure}[htbp]
    \centering
    \begin{lstlisting}[language=Python]
    # Prepare the BitFlipError
    e = QError("Bitflip", loc=[0,0], params=[], p=0.5)

    # Prepare the quantum Markov chain
    qmc = QMarkov(cir_body="qrw.qasm", channels=[e])

    # Prepare model checker and its CFLOBDD backend
    qChecker = fromMarkovModel(qmc)
    qChecker = initWithStr(qChecker, str_list=[])

    # Execution and output
    reachable_dim = qChecker.reachability()
    qChecker.printProjector()
    \end{lstlisting}
    \caption{A Demo for reachability analysis of QMCs. \texttt{channels} is a list of quantum errors and instructions like measurement and reset, containing their occurring positions.}
    \label{fig:demo}
\end{figure}

We implement our symbolic quantum reachability analysis  algorithm with CFLOBDD in a C++ core and provide Python interfaces for invoking. Some of the key components are explained below:

% \begin{itemize}
    \textbf{\textit{Input and output}}. QReach accepts a composite input specification to represent a QMC. A quantum program written in the QASM format \cite{cross2022openqasm} is parsed as the main circuit body of the QMC. Some specified quantum errors, measurements, and other non-unitary channels can be coded as a supplement in the \texttt{Qchannel} type. Once results are obtained, some subspace characters (e.g. dimensions and support vectors) can be output.
    
    \textbf{\textit{Simulation}}. A well-formed toolkit Quasimodo \cite{sistla2023symbolic} based on CFLOBDD was implemented for quantum algorithm simulation. Simulation for quantum circuits, or in other words, applying a sequence of quantum gates on a pure state, is an essential process during the reachability analysis. 
    Therefore, we referred to some of the codes’ implementation details from Quasimodo. Specifically, we adopted Quasimodo's Pybind architecture, which links C++ APIs and Python. However, a simulation framework like Quasimodo cannot handle situations in reachability analysis like mixed states and super-operators. We fixed these issues, added some new features, and stored gate sequences and intermediate projectors to make the simulation execution process not just sequential.

    These methods are covered in \texttt{fromMarkovModel()} and \texttt{reachability()}, while still available to invoke independently for other purposes. For example, \texttt{qchecker.u3()} conducts normal U3 gate simulation on the state vector in the current workspace; \texttt{setProjector()} and \texttt{applyProjector()} methods provide data manipulation in \textit{Projector} and \textit{State vector} as shown in Figure-\ref{fig:arch}. Detailed techniques used in our CFLOBDD simulator different from that in previous works \cite{sistla2022cflobdds,sistla2023symbolic} will be discussed in Section \ref{sec:3.1}.

    % We used a similar but broader simulation framework in Qreach as part of our quantum reachability analysis. Detailed techniques will be discussed in Section \ref{sec:3.1}. In the state preparation step, QReach supports multiple types of initializations, including superpositions and mixed states, which is another addition to Quasimodo. 
    % These methods are covered in \texttt{fromMarkovModel()} and \texttt{qChecker.reachability()}, while still available to invoke independently for other purposes. For example, \texttt{qchecker.u3()} conducts normal U3 gate simulation on the state vector in the current workspace; \texttt{setProjector()} and \texttt{applyProjector()} methods provide data manipulation in \textit{Projector} and \textit{State vector} as shown in Figure-\ref{fig:arch}.
    
    \textit{\textbf{Reachability analysis}}. The efficient algorithm for quantum reachability analysis to be elaborated in Section \ref{sec:D} has been implemented in QReach. We also implemented the interfaces for some of mathematical tricks in \cite{yu2012reachability} (e.g. Choi matrix transformation and maximally entangled state preparation) on CFLOBDDs, which will be critical in a future extension of QReach for computing reachability probabilities (rather than subpaces) of QMCs.  
% \end{itemize}

\begin{figure}[htbp]
    \centering
    \includegraphics[width=0.8\textwidth]{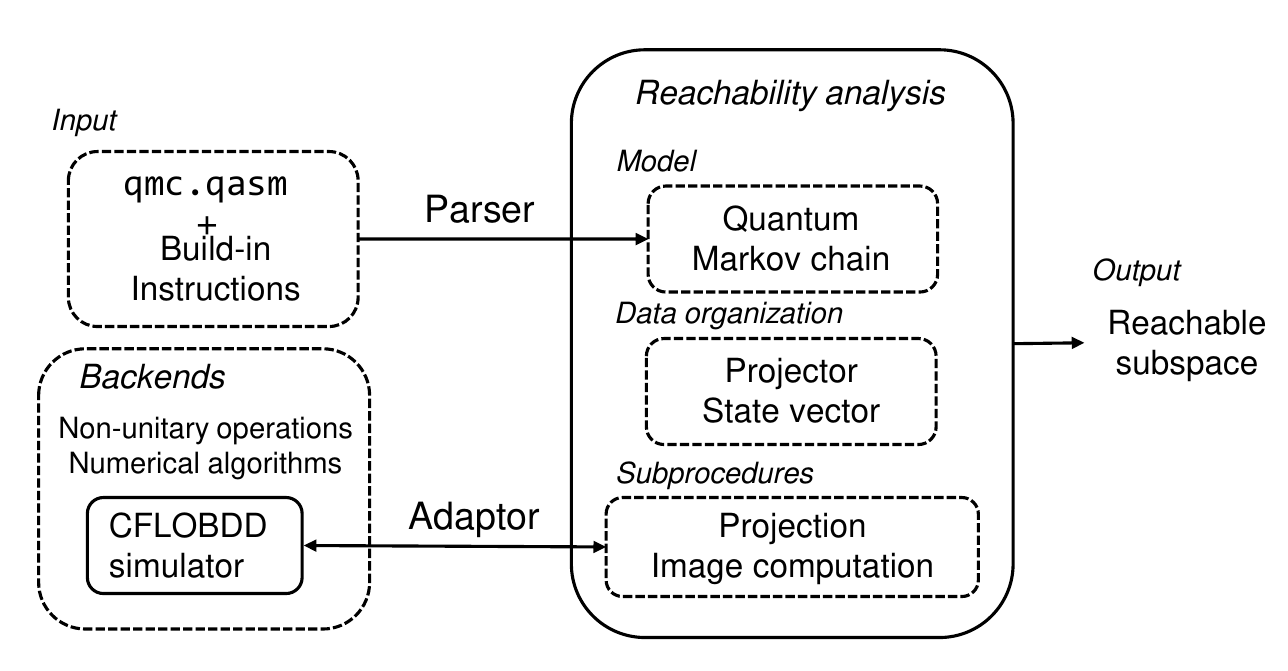}
    \caption{Architecture of QReach.}
    \label{fig:arch}
\end{figure}

% Some shortcomings, numerical algorithms like decomposition. Initial state. Use other backends

\subsection{CFLOBDD for quantum reachability analysis}\label{sec:3.1}

Now we introduce our CFLOBDD backend, which implements some support for numerical algorithms and quantum instructions. In particular, we illustrate how to expand simulation of a quantum system  to its reachability analysis.

As a newly proposed DD-based structure, CFLOBDD attracts our attention due to its distinctive features compared to other DDs for quantum systems. 
%In a tensor-decision diagram (TDD), qubits in a circuit are seen as tensor indices, assigned to a number zero or one. Such binary structure is a direct counterpart of the BDD, while mapping to complex numbers causes reduction and canonicity problems. 
CFLOBDD adopts a single-entry, multi-exit, non-recursive, hierarchical finite-state machine architecture \cite{sistla2022cflobdds}. From a programming perspective, a certain form of procedure call is invoked, leading to some exponential compression over BDDs. Following the name \textquotedblleft context-free language\textquotedblright, the incoming and outgoing edges of groupings are matched according to certain principles. Figure-\ref{fig:cflobdd} provides a general insight into how the edges of CFLOBDDs are matched. The representing capability of the CFLOBDD is exploited in our tool QReach for symbolic reachability analysis of QMCs.

\begin{figure}[htbp]
    \centering
    \includegraphics[width=0.4\textwidth]{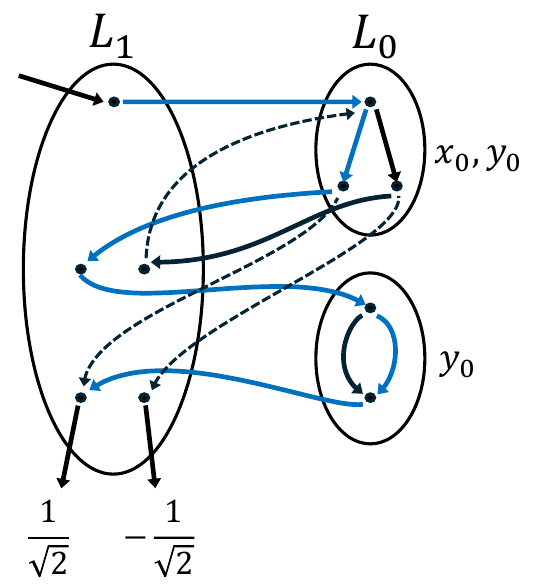}
    \caption{CFLOBDD for Hadamard gate. Indices are represented in the $L_0$ groupings consisting of fork nodes and don't care nodes. A path from the entry of the topmost grouping to the terminal values denotes an assignment to all indices. For example, the bold blue path corresponds with the $\{x_0=0, y_0=1\}$ entry of the Hadamard matrix, resulting a value $\frac{1}{\sqrt{2}}$.}
    \label{fig:cflobdd}
\end{figure}

% Talk about difficulties, why simulation but not representing circuits

Like any other type of DDs, the compression capability of CFLOBDDs only stands out in specific instances. In these cases, canonical reduced forms reuse parts in a DD and save storage from the raw data. However, in general circumstances, the number of nodes required to represent a large-scale matrix or vector is still exponential. Reordering strategies for reduced ordered BDDs to optimize storage usage are hard (NP-complete) \cite{bollig1996improving}. This problem becomes even more severe in algebraic DDs \cite{bahar1997algebric}, where non-Boolean values make it harder to find similar structures in a diagram. Therefore, we cannot simply represent a quantum operation or a projector as a single CFLOBDD without a partition strategy. Unlike classical symbolic model checking, a quantum circuit is usually difficult to partition due to entanglements. We chose an alternative in the QReach backend: using an augmented simulation method to calculate quantum operations. Thus, only state vectors and single quantum gates need to be stored, rather than the whole matrix.

In particular, circuits in QMCs are usually complicated, involving noises and dynamic operations. To handle them, we strengthen CFLOBDD with the following techniques:

% \textit{Basic matrix-vector operations}. Existing methods of basic operations in CFL-\\OBDD and Quasimodo are inherited in QReach.

% \textit{Non-unitary operators}. Different from measure operations in a quantum circuit simulator, the states after measurement are what we need in quantum model checking. (Noises and measurements)

% \textit{Partial trace}. Partial trace is an important class of non-unitary quantum operations. Some built-in instructions of quantum devices, like qubits reset, are based on calculating the partial trace.

\textbf{\textit{Non-unitary operators}}. Apart from normal quantum gates like Hadamard, Pauli, and generic U3 rotation gates, we specifically support two-dimensional matrices of any form, covering those non-unitary operators in quantum noises and measurements; for instance, amplitude damping channels  with operators:
$$E_0 = \begin{bmatrix}
    1 & 0\\
    0 & \sqrt{1-\gamma}
\end{bmatrix},\qquad E_1 = \begin{bmatrix}
    0 & \sqrt{\gamma}\\
    0 & 0
\end{bmatrix}$$
Another example of non-unitary operators is Z-basis measurements. The post measurement states are obtained by applying $P_0 = |0\rangle\langle0|$ and $P_1 = |1\rangle\langle1|$ respectively, followed by normalizations.

\textbf{\textit{Basic methods extensions}}. Some operations are added to CFLOBDD as a basis for top-level algorithms, incorporating the normalization and conjugate transpose. In addition to these methods, an optimizing trick for complex number representation is applied. We used a simplified version of the method proposed in \cite{zulehner2019efficiently}, constructing a Hash function and unique table for complex numbers.

\textbf{\textit{Numerical algorithms}}. Numerical algorithms are critical in QReach. Based on them, some operations that are particularly important in modelling quantum systems (e.g. partial trace and Choi matrix) are now available in QReach (see Example-\ref{example pt}). A key methodology is to decompose operands of calculations into base vectors, replacing complicated operations with matrix-vector multiplication or inner products of vectors. In Section \ref{sec:D}, the high-level description of reachability analysis algorithm also embodies this idea.

\begin{example}[Partial trace]
    \label{example pt}
    Consider a quantum system composed of qubits $A$ and $B$ in state $\rho_{AB}$. Then the state of $A$ can be described by the partial trace operator:
    $$\rho_A \equiv tr_B(\rho_{AB}):=\sum_i(I_A\otimes \langle i|_B)\rho_{AB}(I_A\otimes |i\rangle_B)$$
    For simplicity, suppose the system is in a pure state $|\psi\rangle = |0\rangle|\lambda\rangle + |1\rangle|\mu\rangle$. Tracing out qubit $A$, qubit $B$ should be in the mixed state $\rho = |\lambda\rangle\langle\lambda| + |\mu\rangle\langle\mu|$. In QReach, this procedure is conducted in the following steps to avoid redundant matrix manipulations and adjustments to CFLOBDD's internal structures: (1) Perform a Z-basis measurement on $A$ and get unnormalized post measurement states $|0\rangle|\lambda\rangle$ and $|1\rangle|\mu\rangle$; (2) Apply $X$ gate to $A$ conditionally on the measurement result one; (3) Let the collection $S = \{|0\rangle|\lambda\rangle, |0\rangle|\mu\rangle\}$. Then $S$ can be viewed as $\rho$'s representation and participate in later calculations. In some cases, like reachability analysis, the norm of a state vector is not essential and could be omitted thereby. Note that after these operations qubit $A$  remains in a tensored zero state, because the number of qubits in a CFLOBDD is required to be an exponential power of 2. We apply an $X$ gate on the measure-one result to make the effect looks  like resetting a qubit.
\end{example}

% \begin{example}[Choi matrix]
%     \label{example choi}
%     The Choi matrix is another representation of super operators.
% \end{example}

% advantages about CFLOBDD
% Compared to weighted reduced DDs (e.g. TDD or WCFLOBDD), complex numbers in a CFLOBDD are organized at the terminal entry but not on the internal edges. This feature decouples the data structure and coefficients, allowing us to operate certain parts of them independently.

\section{Algorithm for reachability analysis}
\label{sec:D} The existing algorithm for reachability analysis of QMCs is based on Choi matrix representation of quantum operations introduced in \cite{yu2012reachability}. In this section, we propose a more efficient algorithm for the same purpose (see 
Algorithm-\ref{alg-compute-reachable}).  
%is the main reachability analysis procedure in QReach. 

Our algorithm is a natural extension of reachability analysis in classical model checking using a BFS-based technique. The main difference is that we are dealing with reachable \textit{subspaces} of the Hilbert space $\mathcal{H}$ rather than \textit{subsets} of a finite set of states in the classical case. Therefore, Algorithm-\ref{alg-compute-reachable} traverses each possible \textit{dimension} of a finite-dimensional Hilbert space non-repetitively rather than each reachable state as in the classical case.
%is intuitive and efficient enough for such model-checking problems. 
Note that the code segment from Line 7 to Line 12 is the process of extracting vectors orthogonal to those that have been searched, which is similar to frontier set simplification in classical symbolic model checking. 

The nontrivial subprocedures in Algorithm-\ref{alg-compute-reachable} differing from that in classical reachability analysis are the projection and image computation (Line 5 and Line 7). To reduce the representing and temporal cost in the algorithm, our basic idea is to take eigenvectors into calculation instead of the whole matrix. In practice, most of the projections are low-rank, which ensures the efficiency of this idea.
\begin{algorithm} 
	\caption{Computing reachable space} 
	\label{alg-compute-reachable} 
	\begin{algorithmic}[1]
		\REQUIRE Super operator $\mathcal{E}$ in Hilbert space $\mathcal{H}$ with dimension $d$; set of initial states $P$
		\ENSURE A set of orthogonal basis $P'$ of reachable subspace of $\mathcal{H}$
            \STATE $P' \gets Gram\_Schmidt(P),\quad cnt\gets size(P')$
            \STATE Initialize queue $Q$ with $P'$
            \WHILE{$Q$ not empty and $cnt < d$}
            \STATE $curr\_state \gets$ $Q$.pop()
            \STATE $expanded\_states \gets$ $\texttt{Image}(\mathcal{E}, curr\_state)$
            \FOR{$s$ in $expanded\_states$}
            \STATE $s\gets s - \texttt{Project}(P', s)$
            \STATE $s\gets \texttt{normalize}(s)$
            \IF{$s$ is not zero vector}
            \STATE $Q$.push($s$)
            \STATE $P'$.append($s$)
            \STATE $cnt\gets cnt+1$
            \ENDIF
            \ENDFOR
            \ENDWHILE
		\RETURN $P'$ 
	\end{algorithmic}
\end{algorithm}

\textbf{Implementation in QReach}. For image computation, we exploit the simulation functionality of our backend data structure  CFLOBDD. The runtime of the backend's simulation highly determines our algorithm's efficiency. In this step, non-unitary operations like noises, measurements, qubit resets, and deallocations will be simulated by the augmented simulator introduced in the past section. All the simulations of channels together make up the image computation of Kraus representations $\mathcal{E}(\rho) = \sum_i E_i \rho E_i^\dag$.

A projector onto a subspace of the Hilbert space can be represented by a set of orthogonal support vectors of the subspace. This technique can be viewed as a quantum version of partitioning, which usually appears as forms of disjunctions and conjunctions in classical symbolic model checking \cite{burch1994symbolic}. Formally, let $\mathcal{P}$ be the projector onto a subspace with an orthonormal basis $\{|i\rangle\}$, that is, $\mathcal{P} = \sum |i\rangle\langle i|$, then we set $P$ to be the set of $\ket{i}$'s, and
$$\texttt{Project}(P, \ket{s}) = \sum \langle s|i\rangle^\ast |i\rangle$$
We conduct conjugate-transposing on $\ket{s}$ instead of $\ket{i}$ to invoke vector operations as few as possible. The computational complexity of subprocedure \texttt{Image} and \texttt{Project} are both $O(d^2)$ given a constant number of Kraus matrices. 

The following theorem shows the correctness and complexity of our algorithm.

\begin{theorem}
\label{correctness}
The output $P'$ and input $P$ of Algorithm-\ref{alg-compute-reachable} satisfies:
$\text{span}(P')=\mathcal{R}_\mathcal{C}(\rho)=\bigvee^{d-1}_{i=0} \text{supp}(\mathcal{E}^i(\rho)),$
where $\rho$ is the initial state, and $\text{supp}(\rho)=\text{span}(P)$. The time complexity of Algorithm-\ref{alg-compute-reachable} is $O(d^3)$
\end{theorem}
\begin{proof}
    % See Appendix-\ref{appendix1}
    Following the theorem 1 in \cite{yu2012reachability}), for $d = dim(\mathcal{H})$, and any density operator $\rho$ in $\mathcal{H}$,
    \begin{align}\label{lma}
        \mathcal{R}_\mathcal{C}(\rho) = \text{supp}\left(\sum_{i=0}^{d-1}\mathcal{E}^i(\rho)\right)
    \end{align}
    And all reachable states can be reached in at most $d$ iterations. The correctness is proved in three steps: \begin{itemize}
        \item[i)] The main \textit{while} loop terminates in at most $d$ steps;
        \item[ii)] When terminates, $\text{span}(\texttt{Image}(\mathcal{E}, P')) = \text{span}(P')$;
        \item[iii)] $\mathcal{R}_\mathcal{C}(P') = \mathcal{R}_\mathcal{C}(\rho)$.
        %At the beginning of each loop, it holds that $\mathcal{R}_\mathcal{C}(\texttt{Image}(\mathcal{E},Q))\lor \text{supp}(P')=\mathcal{R}_\mathcal{C}(\rho)$.
    \end{itemize}

    At last, combining the complexity of broad-first-search and subprocedures, the Algorithm-\ref{alg-compute-reachable} has complexity $O(d^3)$, which is an improvement over  $O(d^{4.7454})$ in \cite{yu2012reachability}.
\end{proof}

The dimension of a Hilbert space often grows exponentially larger in quantum systems. Therefore, our algorithm will inevitably become inefficient on quantum circuits with more than twelve qubits. Although limited on the scale of quantum systems, Algorithm-\ref{alg-compute-reachable} is stable for the number of quantum operations and the dimension of initial spaces. The BFS strategy and the frontier set simplification ensure that only dimensions that are reached for the first time can be counted.

% add something about the numerical algorithms on CFLOBDD

\section{Use cases and experiments}
% \begin{itemize}
%     \item Quantum random walk
%     \item repeat-until-success circuits
% \end{itemize}
Some cases are studied in this section, providing insight into practical applications of our tool QReach for quantum reachability analysis in the future. All experiments were conducted on a personal computer with hardware configurations: Intel i5-13600kf CPU with 14 cores; 32GB RAM. The experimental results are presented in Table-\ref{tab:result}.

\begin{table}
\centering
\caption{Experimental results. The noise in QRW is amplitude dumping. For RUS, the circuits implement $(I+2iZ)/\sqrt{5}$, $(2X+\sqrt{2}Y+Z)/\sqrt{7}$, and $(3I+2iZ)/\sqrt{13}$ respectively.}\label{tab:result}
\begin{tabular}{|c|c|c|c|c|c|c|}
\hline
 & \#Qubits & Ope. type & Initial dim. & Time(s) & \#Edges & Reachable dim.\\
\hline
\multirow{4}{*}{Grover} & 7 & Unitary & 1 & 0.0252 & 268 & 2 \\
 & 15 & Unitary & 1 & 0.0303 & 350 & 2 \\
 & 31 & Unitary & 1 & 0.0516 & 432 & 2 \\
 & 63 & Unitary & 1 & 0.0885 & 514 & 2 \\
 \hline
\multirow{7}{*}{QRW} & 3 & Unitary & 1 & 0.0284 & 435 & 6 \\
 & 5 & Unitary & 1 & 0.0561 & 1337 & 10 \\
 & 7 & Unitary & 1 & 0.196 & 7512 & 34 \\
 & 9 & Unitary & 2 & 379.64 & 608097 & 512 \\
 & 7 & Noise & 2 & 0.446 & 10211 & 64 \\
 & 9 & Noise & 2 & 110.70 & 447811 & 512 \\
 & 10 & Noise & 2 & 117.36 & 421038 & 1024 \\
 \hline
\multirow{3}{*}{RUS} & 3 & Measure & 1 & 0.0262 & 130 & 2 \\
 & 2 & Measure & 1 & 0.0216 & 74 & 2 \\
 & 2 & Measure & 1 & 0.0203 & 60 & 1 \\

\hline
\end{tabular}
\end{table}

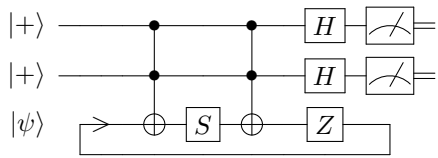
\begin{figure}[htbp]\centering
\normalsize\begin{equation*}\qquad 
\Qcircuit @C=.8em @R=.5em {
\lstick{\ket{+}}    & \qw  & \qw & \qw & \ctrl{1} & \qw      & \ctrl{1} & \qw & \gate{H}  & \meter & \cw \\
\lstick{\ket{+}}    & \qw  & \qw & \qw & \ctrl{1} & \qw      & \ctrl{1} & \qw & \gate{H}  & \meter & \cw \\
\lstick{\ket{\psi}} &      & \qw> & \qw & \targ    & \gate{S} & \targ    & \qw & \gate{Z} & \qw \\
                    & \qwx & \qw & \qw & \qw      & \qw      & \qw      & \qw & \qw      & \qw \qwx
}
\end{equation*}
\caption{A repeat-until-success circuit~\cite{sequential,10.5555/2685179.2685181} for gate $V_3=(I+2iZ)/\sqrt{5}$.}
\label{fig:rus}
% \hrulefill \vspace*{4pt}
\end{figure}

\textbf{\textit{Grover search}}. The Grover search algorithm provides a quadratic speedup over a series of classical search algorithms \cite{grover1998quantum}. The main idea of the Grover search is to apply a quantum subroutine iteratively which leads to a rotation from the initial state to the target state. %Although the original Grover search is a combinational quantum circuit with a finite length that doesn't satisfy the rigorous definition of the quantum Markov chain, as an example, we could still verify a temporal property by 
QReach's reachability analysis under some float precision shows that during iterations, the state is always located in the 2-dimensional subspace spanned by the initial state and the target state.

\textbf{\textit{Quantum random walk}}. We tested quantum random walk circuits (QRW) with different numbers of qubits which have a similar structure with Example-\ref{example1}. To demonstrate more functionalities of QReach, mixed initial states and amplitude dumping noises are introduced in front of the Hadamard gate. The (dimension of) reachable space computed by QReach for these quantum circuits are given in Table-\ref{tab:result}. 

\textbf{\textit{Repeat-until-success circuits}}. Repeat-until-success (RUS) circuits \cite{10.5555/2685179.2685181} are a type of circuit that decides whether to repeat or terminate based on the measurement results. It is usually used to design circuits with fewer non-Clifford gates or ancilla qubits. In QReach, it can be modelled as a quantum Markov chain with measurements and qubit resetting as parts of the channel. We tested some of the examples in \cite{10.5555/2685179.2685181}. It is clear that the reachable dimensions should be 2 or 1, depending on whether the resulting quantum states differ by only one global phase if the measurement succeeds or fails.

There is a consensus that every future tool released in quantum model checking must face the problem of finding broader applications. Besides these cases, we are improving the scope of QReach and exploring more possible applications on sequential quantum circuits and protocols.

%\section{Conclusion}


\newpage

\bibliographystyle{splncs04}
\bibliography{reference}

\newpage

\appendix
\section{Proof of Theorem-\ref{correctness}}\label{appendix1}
\subsection{Preliminaries}
Some notations in Theorem-\ref{correctness} and other parts of this paper are formally defined below. For more information, please refer to \cite{ying2021model}.
\begin{definition}
    The support $\text{supp}(\rho)$ of a density operator $\rho\in \mathcal{D}(\mathcal{H})$ is the subspace of $\mathcal{H}$ spanned by the eigenvectors of $\rho$ with non-zero eigenvalues.
\end{definition}
\begin{definition}
    Let $\{X_k\}$ be a family of subspaces. The join of $\{X_k\}$ is defined by
    $$\bigvee_k X_k = \text{span}\left(\bigcup_k X_k\right).$$
    In particular, $X\lor Y$ means the join of subspaces $X$ and $Y$.
\end{definition}

The following Lemma (Theorem 1 from \cite{yu2012reachability}) is the closed form characterization of the reachable space.
\begin{lemma}\label{lamma1}
    For $d = dim(\mathcal{H})$, and any density operator $\rho$ in $\mathcal{H}$,
    \begin{align}\label{eq_lemma}
        \mathcal{R}_\mathcal{C}(\rho) = \text{supp}\left(\sum_{i=0}^{d-1}\mathcal{E}^i(\rho)\right)
    \end{align}
\end{lemma}
First, by definition, the reachable subspace $\mathcal{R}_\mathcal{C} = \bigvee^{\infty}_{i=0} \text{supp}(\mathcal{E}^i(\rho))$. Following the fact that $\text{supp}(\rho+\sigma) = \text{supp}(\rho)\lor \text{supp}(\sigma)$, the right hand side of \ref{eq_lemma} can be written as $\bigvee^{d-1}_{i=0} \text{supp}(\mathcal{E}^i(\rho))$ denoted by $X$. It is obvious that $X \subseteq \mathcal{R}_\mathcal{C}(\rho)$. To prove the inversion, let subspaces $Y_n := \text{supp}\left(\sum_{i=0}^{n}\mathcal{E}^i(\rho)\right)$. And we have the relation $Y_0\subseteq Y_1\subseteq\cdots\subseteq Y_n\subseteq\cdots$. Suppose $r$ is the smallest integer such that $Y_r = Y_{r+1}$. According to the observation that $Y_{n+1} = \text{supp}(\rho+\mathcal{E}(P_n))$, where $P_n$ is the projector on $Y_n$, $Y_n = Y_r$ for any $n>r$ by induction. Since $\text{dim}(Y_r)$ is at most $d$, and $\text{dim}(Y_0)\geq 1$, there must be $r<d$. Finally, $\mathcal{R}_\mathcal{C} = \bigvee^{\infty}_{i=0} \text{supp}(\mathcal{E}^i(\rho)) \subseteq \bigvee^{d-1}_{i=0} \text{supp}(\mathcal{E}^i(\rho)) = X$.

\subsection{Correctness}
We are going to prove our theorem in three steps: \begin{itemize}
    \item[i)] The main \textit{while} loop terminates in at most $d$ steps;
    \item[ii)] When terminates, $\text{span}(\texttt{Image}(\mathcal{E}, P')) = \text{span}(P')$;
    \item[iii)] $\mathcal{R}_\mathcal{C}(P') = \mathcal{R}_\mathcal{C}(\rho)$.
    %At the beginning of each loop, it holds that $\mathcal{R}_\mathcal{C}(\texttt{Image}(\mathcal{E},Q))\lor \text{supp}(P')=\mathcal{R}_\mathcal{C}(\rho)$.
\end{itemize}
Here we abuse the notation $\mathcal{R}_\mathcal{C}(P) := \bigvee_{s \in P}\mathcal{R}_\mathcal{C}(s)$. The \texttt{Image} function is of the same meaning with that in Algorithm-\ref{alg-compute-reachable}. It means that for each element in $P'$, compute its image under $\mathcal{E}$ and return the image set.

i). Suppose the number of initial elements in $Q$ is $s$, and the number of loops is $l$; for the $i$th loop, the number of searched vectors in Line 9 is $k_i$. In the $i$th loop, the variables change with $cnt\leftarrow cnt+k_i$, $|Q|\leftarrow |Q|+(k_i-1)$. According to the termination conditions, we have the following two inequalities:
$$\begin{cases}
    &\sum_{i=1}^l k_i - l + s \geq 0\\
    &\sum_{i=1}^l k_i + s \leq d
\end{cases}$$
And we have $l\leq d$ by combining them.

ii). Note that $Q$ is always a subset of $P'$. Any elements will be simultaneously added to $P'$ and $Q$. Any time when an element $s$ is deleted from $Q$, it results in that $P'$ is expanded with the image of $s$. When terminates, if the condition $cnt \geq d$ is triggered, $\text{span}(P') = \mathcal{H} = \text{span}(\texttt{Image}(\mathcal{E}, P'))$. Otherwise, $Q$ is empty, which means that images of any elements in $P'$ are included in $\text{span}(P')$.

iii). $P'$ is initialized with a set of orthogonal eigenvectors of $\rho$. And the update of $P'$ only involves with image computation and linear combination.

The above three propositions imply that when the program terminates, $$\text{span}(P') = \mathcal{R}_\mathcal{C}(P') = \mathcal{R}_\mathcal{C}(\rho).$$

\subsection{Complexity and scalability of Algorithm-\ref{alg-compute-reachable}} 
Combining Theorem-\ref{correctness} and the computational complexity of subprocedures, the Algorithm-\ref{alg-compute-reachable} has complexity $O(d^3)$, which is an improvement over  $O(d^{4.7454})$ in \cite{yu2012reachability}. The dimension of a Hilbert space often grows exponentially larger in a quantum system. Therefore, our algorithm to give a whole description of reachable spaces will inevitably become inefficient on quantum circuits with more than twelve qubits. To remedy this issue, we will explore approximation algorithms and other backends for QReach to improve its flexibility in experiments in our future work.

Although QReach is limited on the scale of quantum systems, Algorithm-\ref{alg-compute-reachable} is stable for the number of quantum operations and the dimension of initial spaces. The BFS strategy and the frontier set simplification ensure that only dimensions that are reached for the first time can be counted. In potential real model-checking scenarios, different noises, measurements, and complicated state initializations will be common. We can also imagine that there are relatively few cases where a searched reachable space covers most dimensions because noteworthy subspaces (e.g., a 'bad state') may be small.

\end{document}